\documentclass[11pt]{article}

\usepackage[utf8]{inputenc}
\usepackage{newunicodechar}
\usepackage[T1]{fontenc}
\usepackage{lmodern}

\usepackage{amsmath,amssymb,amsfonts,mathtools,amsthm}

\newtheorem{theorem}{Theorem}[section]
\newtheorem{lemma}[theorem]{Lemma}

\usepackage[linesnumbered,vlined]{algorithm2e}
\DontPrintSemicolon

\newcommand{\newFunc}[2]{\newcommand{#1}{\FuncSty{#2}\xspace}}

\newcommand{\concept}[1]{\textbf{#1}}

\newcommand{\etal}{\textit{et al.}}

\newFunc{\Read}{read}
\newFunc{\Write}{write}
\newFunc{\ShiftLeft}{shiftLeft}
\newFunc{\ShiftRight}{shiftRight}

\newcommand{\Set}[1]{\left\{{#1}\right\}}

\newcommand{\Tuple}[1]{\left\langle{#1}\right\rangle}

\DeclarePairedDelimiter{\card}{\lvert}{\rvert}

\DeclarePairedDelimiter{\parens}{(}{)}

\newunicodechar{¬}{\neg}   
\newunicodechar{∧}{\wedge} 
\newunicodechar{⊕}{\oplus} 
\newunicodechar{∨}{\vee}   
\newunicodechar{→}{\rightarrow}  
\newunicodechar{⇒}{\Rightarrow}  
\newunicodechar{←}{\leftarrow}   
\newunicodechar{↔}{\leftrightarrow}   
\newunicodechar{⇔}{\Leftrightarrow}   
\newunicodechar{⊢}{\vdash}   

\newunicodechar{⊗}{\otimes} 

\newunicodechar{↦}{\mapsto}  
\newunicodechar{↪}{\hookrightarrow}  
\newunicodechar{↣}{\rightarrowtail}  

\newunicodechar{∀}{\forall} 
\newunicodechar{∃}{\exists} 

\newunicodechar{∆}{\Delta} 
\newunicodechar{∇}{\nabla} 

\newunicodechar{∩}{\cap}    
\newunicodechar{∪}{\cup}    
\newunicodechar{∖}{\setminus}    
\newunicodechar{×}{\times}    
\newunicodechar{∈}{\in}     
\newunicodechar{∉}{\not\in}     
\newunicodechar{∋}{\ni}     
\newunicodechar{⊆}{\subseteq}    
\newunicodechar{⊇}{\supseteq}    
\newunicodechar{⊑}{\sqsubseteq}    
\newunicodechar{⊒}{\sqsupseteq}    
\newunicodechar{≾}{\lesssim}
\newunicodechar{≿}{\grtsim}
\newunicodechar{⪯}{\preceq}
\newunicodechar{⪰}{\succeq}

\newunicodechar{∘}{\circ}    
\newunicodechar{⋅}{\cdot}    

\newunicodechar{∅}{\emptyset}  
\DeclareUnicodeCharacter{2115}{\mathbb{N}} 
\DeclareUnicodeCharacter{2124}{\mathbb{Z}} 
\DeclareUnicodeCharacter{211A}{\mathbb{Q}} 
\DeclareUnicodeCharacter{211D}{\mathbb{R}} 
\DeclareUnicodeCharacter{2102}{\mathbb{C}} 
\newunicodechar{ℵ}{\aleph}     
\newunicodechar{ℶ}{\beth}      

\DeclareUnicodeCharacter{2131}{\mathcal{F}} 

\newunicodechar{∑}{\sum}  
\newunicodechar{∏}{\prod} 
\newunicodechar{∫}{\int}  
\newunicodechar{±}{\pm}   

\newunicodechar{≠}{\neq}    
\newunicodechar{≡}{\equiv}  
\newunicodechar{≈}{\approx} 
\newunicodechar{≃}{\simeq}  
\newunicodechar{≅}{\cong}   
\newunicodechar{≤}{\le}     
\newunicodechar{≥}{\ge}     
\newunicodechar{∣}{\divides}

\newunicodechar{⊤}{\top}    
\newunicodechar{⊥}{\bot}    
\newunicodechar{∞}{\infty}  
\newunicodechar{§}{\S}      
\newunicodechar{□}{\Box}    
\newunicodechar{◇}{\Diamond}      

\newunicodechar{Γ}{\Gamma}  
\newunicodechar{Δ}{\Delta}  
\newunicodechar{Θ}{\Theta}  
\newunicodechar{θ}{\theta}  
\newunicodechar{Λ}{\Lambda} 
\newunicodechar{Ξ}{\Xi}     
\newunicodechar{Π}{\Pi}     
\newunicodechar{Σ}{\Sigma}  
\newunicodechar{Φ}{\Phi}    
\newunicodechar{Ω}{\Omega}  
\newunicodechar{α}{\alpha}  
\newunicodechar{β}{\beta}   
\newunicodechar{γ}{\gamma}  
\newunicodechar{δ}{\delta}  
\newunicodechar{ε}{\epsilon} 
\newunicodechar{ζ}{\zeta}   
\newunicodechar{η}{\eta}    
\newunicodechar{ι}{\iota}   
\newunicodechar{κ}{\kappa}  
\newunicodechar{λ}{\lambda} 
\newunicodechar{μ}{\mu}     
\newunicodechar{ξ}{\xi}     
\newunicodechar{π}{\pi}     
\newunicodechar{ρ}{\rho}    
\newunicodechar{σ}{\sigma}  
\newunicodechar{τ}{\tau}    
\newunicodechar{φ}{\phi}    
\newunicodechar{χ}{\chi}    
\newunicodechar{ψ}{\psi}    
\newunicodechar{ω}{\omega}  

\title{The consensus number of a shift register\\equals its width}
\author{James Aspnes\\Yale University}

\date{\quad}

\begin{document}

\maketitle

\begin{abstract}
    The consensus number of a $w$-bit register supporting logical
    left shift and right shift operations is
    exactly $w$, giving an example of a class of types, widely
    implemented in practice, that populates all levels of
    the consensus hierarchy. This result generalizes to $w$-wide shift registers
    over larger alphabets.
    In contrast,
    a register providing arithmetic right shift, which replicates the
    most significant bit instead of replacing it with zero,
    is shown to solve consensus for any fixed number of processes as long as
    its width is at least two.
\end{abstract}

\section{Introduction}
\label{section-introduction}

Herlihy's consensus
hierarchy~\cite{Herlihy1991waitfree,Jayanti1997}
classifies each shared-memory object type $T$ by its \concept{consensus number},
the maximum number of processes $n$ for which it is possible to solve
wait-free consensus using an unbounded number of objects of type $T$
along with an unbounded number of read/write registers.
We show that logical shift registers of varying width cover all finite
levels of the consensus hierarchy, by showing that $w$-bit registers
have consensus number $w$.

The remainder of this section provides a brief background on the shared-memory
model and consensus numbers. Section~\ref{section-shift-registers}
formally defines the class of objects under consideration.
Section~\ref{section-logical-shift-registers} demonstrates that the consensus
number of a $w$-bit logical shift register is shown to be exactly $w$, using a
characterization of consensus numbers due to Ruppert~\cite{Ruppert2000}; this result also applies to
$w$-wide logical shift registers over larger alphabets.
Section~\ref{section-arithmetic-shift-registers} considers
arithmetic shift registers, where right-shift
operations replicate the leftmost bit instead of replacing it with a
zero, and shows that these objects have infinite consensus number for
any width $w≥2$. Section~\ref{section-conclusions} gives a brief
discussion of these results and possible future work.

\subsection{Shared-memory model}

We consider the standard asynchronous shared-memory model, where a
collection of $n$ \concept{processes} interact by operating on shared
\concept{objects}. Both the processes and objects are represented as
finite-state machines, and a \concept{configuration} of the system
specifies the states of all processes and objects. In each
configuration, each process has a pending \concept{operation} on some
object. An \concept{adversary scheduler} chooses which of the $n$
pending operations to execute next. Executing an operation updates the
state of the object to which it is applied, and returns a value that
is used to update the state of the executing process.

A standard object type is an \concept{read/write register}, which
supports a $\Read$ operation that returns its current state, and, for
each value $v$, a $\Write(v)$ operation that updates its state to
$v$ and returns nothing.

Other object types may have more complicated behavior.
A fundamental problem in the theory of shared-memory systems is
classifying the relative strength of different object types.

\subsection{Consensus numbers}

\concept{Consensus} allows each of $n$ processes to provide
an input and obtain an common output after a finite number of steps.
Formally, the requirements are
\begin{description}
    \item[Agreement] All processes obtain the same output.
    \item[Termination] Each process finishes in a finite number of
        steps.
    \item[Validity] Each process obtains an output equal to some
        process's input.
\end{description}
A protocol is \concept{wait-free} if termination holds in all possible
schedules without any fairness constraints.

The \concept{consensus number}~\cite{Herlihy1991waitfree,Jayanti1997}
of a shared-memory object $T$ is the largest number of processes $n$
for which it is possible to solve wait-free consensus using an
unlimited number of $T$ objects supported by an unlimited number of
read/write registers. An object has consensus number $∞$ if it implements
consensus for any $n$. 

Consensus numbers create a hierarchy, where objects with
low consensus numbers cannot implement objects with higher consensus
numbers. This hierarchy is fully populated in the sense that for each
possible $m ∈ \Set{1,\dots,∞}$, there is an object with consensus
number $m$. Some parameterized families of objects that demonstrate this
property are $m$-restricted consensus objects, which provide an
operation that satisfies the requirements of consensus, but only to
the first $m$ processes that invoke it; and
$m$-sliding read/write registers or
$m$-buffers~\cite{MostefaouiPR2018,EllenGSZ2020}, which are similar to
standard read/write registers except that they return the last $m$ 
values written instead of just the last value written.
These families of objects are of significant theoretical interest, but
have the unfortunate deficiency of not appearing as a common feature 
in practical systems.

Proving that an object $T$ has consensus number $m$ for finite $m$ usually
involves two parts: showing that $T$ plus registers can solve
consensus for $n=m$, and that it cannot solve consensus for $n≥m+1$.
The first part can be demonstrated by giving a working algorithm; the
second usually involves a bivalence argument along the lines of the
classic Fischer-Lynch-Paterson impossibility
result~\cite{FischerLP1985} for asynchronous message passing, which
can be adapted to wait-free shared memory~\cite{LouiA1987}. The
paper of Herlihy introducing the consensus
hierarchy~\cite{Herlihy1991waitfree} gives many
examples of such arguments for specific shared-memory objects.

We will find it convenient to use a general characterization of the
consensus number of certain classes of shared-memory objects given by
Ruppert~\cite{Ruppert2000}, which has the advantage of only needing to
examine a finite class of possible executions on a single object.
We describe this characterization in
Section~\ref{section-logical-shift-registers}.

\section{Shift registers}
\label{section-shift-registers}

A shift register of width $w$ holds a string $x_{w-1} \dots x_0$,
where each $x_i$ is drawn from an alphabet that contains a special
null symbol $0$, and, to avoid triviality, at least one non-null symbol $1$.

Given such a string,
define one-position shift
operations $\ell$ (left shift), $r$ (logical right shift), and $s$ (arithmetic
right shift), written in postfix notation, where the effect of
each operation is given by the rules
\begin{align*}
    x\ell &= x_{w-2} x_{w-3} \dots x_0 \, 0 \\
    xr    &= 0\, x_{w-1} x_{w-2} \dots x_1 \\
    xs    &= x_{w-1} x_{w-1} x_{w-2} \dots x_1 
\end{align*}
These are destructive operations: $\ell$ removes $x_{w-1}$, while $r$
and $s$ remove $x_0$. This means that $\ell$ is not, in general, an
inverse of either $r$ or $s$.

We can also define $k$-fold versions of these operations:
\begin{align*}
    x\ell^k &= x_{w-k-1} \dots x_0 \, 0^k \\
    xr^k    &= 0^k \, x_{w-1} x_{w-2} \dots x_k \\
    xs^k    &= x_{w-1}^k x_{w-1} x_{w-2} \dots x_k 
\end{align*}

Shift operations on bits, even in their $k$-fold variants, are widely
available both as machine-level instructions and as operators in
common programming languages. For example, assuming that any return
value is ignored, $\ell^k$ corresponds to the left-shift operator
$\verb:<<=:k$ in C~\cite{KernighanR1988}, while $r^k$ corresponds to
the right-shift operator $\verb:>>=:k$ when applied to variables with
unsigned types. Applying $s^k$ corresponds to doing an arithmetic
right shift with sign extension. While C does not provide an operation
that is guaranteed to perform an arithmetic right shift, 
they are available at the machine-code level on most CPUs; for
example, the Intel architecture provides them as the Shift Arithmetic
Right (SAR) instruction~\cite{Intel1997}.

We consider two classes of registers supporting shift operations in
addition to $\Write(v)$ and $\Read$. \concept{Logical shift
registers}, examined in Section~\ref{section-logical-shift-registers}, provide 
$\ell^k$ and $r^k$ for all $k≥1$. \concept{Arithmetic shift
registers}, examined in Section~\ref{section-arithmetic-shift-registers},
instead provide $\ell^k$ and $s^k$. We will see that the choice between
$r^k$ and $s^k$ gives very different consensus numbers for these objects.

\section{Logical shift registers}
\label{section-logical-shift-registers}

Here we present the main result of the paper, that logical shift
registers have consensus number equal to their width. To prove this,
we use a characterization of consensus number for types supporting
$\Read$ operations due to Ruppert~\cite{Ruppert2000}.

A readable type $T$ is
\concept{$n$-discerning}~\cite[Definition 13]{Ruppert2000}
if there is some initial state $q^0$ of a
$T$-object, a partition of the processes $p_1,\dots,p_n$ into disjoint
nonempty teams $A$ and $B$, and a choice of update operation $π_i$ on
the object for each process $p_i$, such that the view of any process
in an execution starting in $q^0$ and consisting of operations $π_i$
starting with an operation by team $A$ is distinct
from the view of the same process in any such execution starting with an
operation by team $B$.

This is formalized by defining sets $R_{A,j}$
and $R_{B,j}$, where $R_{A,j}$ consists of all pairs $\Tuple{r,q}$ where
there is an execution $q^0 π_{i_1},π_{i_2,},\dots,π_{i_α}$
where the first operation $π_{i_1}$ is performed by some $p_{i_1}$ in
$A$, at least one of the operations is $p_j$'s
operation $π_j$, $π_j$ returns $r$ in the execution, and the object
is left in state $q$ at the end of the execution.
The set $R_{B,j}$ is defined similarly, the only difference being that $p_{i_1} ∈ B$.
Then a readable
type is $n$-discerning if there exist disjoint nonempty teams $A,B$ and operations
$π_i$ such that $R_{A,j}∩R_{B,j}=∅$ for all $j$.

The usefulness of this definition is that a readable type $T$
can solve $n$-process consensus for $n≥2$ when supplemented by read/write
registers if and only if it is $n$-discerning~\cite[Theorems~15
and~18]{Ruppert2000}. 
We will use this fact to establish the consensus
number of a $w$-bit logical shift register.

To show that a $w$-wide logical shift register is not $n$-discerning
for $n≥w+1$, we will need to show
that any choice of teams $A$ and $B$ and updates $π_i$ will
give $R_{A,j} ∩ R_{B,j} ≠ ∅$ for some $j$. Because no update operation
returns a value, this simplifies to showing that there always exist
executions
$q^0 π_{i_1},π_{i_2,},\dots,π_{i_α}$ and 
$q^0 π_{i'_1},π_{i'_2,},\dots,π_{i'_β}$ 
that leave the object in the same state, overlap in at least one operation $π_{i_a} = π_{i'_b} = π_j$,
and have $p_{i_1} ∈ A$ and $p_{i'_1} ∈ B$.

Our strategy for proving that such bad executions exist 
when $n ≥ w+1$ will be to show a sequence of lemmas
restricting what updates can be used, culminating in a combinatorial
argument showing that we can't have too many of the updates that are
left.

First let us observe that no $π_i$ can be a \Write:
\begin{lemma}
    \label{lemma-no-writes}
    If $R_{A,j} ∩ B_{A,j} = ∅$ for all $j$, then no $π_i$ is a $\Write$.
\end{lemma}
\begin{proof}
    Suppose $π_i$ is a \Write, and assume without loss of generality
    that $p_i∈A$. Let $π_{i'}$ be the update for some $p_{i'}∈B$. Then
    $q^0 π_{i'} π_i = q^0 π_i$, and the executions overlap in $π_i$.
\end{proof}

This leaves only shifts as possible update operations.
We can show that the updates $π_i$ divide neatly into
left shifts and right shifts depending on which team $p_i$ is in:
\begin{lemma}
    \label{lemma-team-direction}
    If $R_{A,j} ∩ B_{A,j} = ∅$ for all $j$, then there are no two
    processes $p_i,p_{i'}$ such that $p_i∈A$, $p_{i'}∈B$, and $π_i$ and $π_{i'}$
    are both left shifts or both right shifts.
\end{lemma}
\begin{proof}
    Suppose $i∈A$, $i'∈B$, $π_i=\ell^k$, and $π_{i'} = \ell^{k'}$ for
    some $k$ and $k'$. Then $q^0 π_i π_{i'} = q^0 \ell^k \ell^{k'}
    = q^0 \ell^{k+k'} = q^0 \ell^{k'} \ell_k = q^0 π_{i'} π_i$.
    Overlap is immediate since the two executions include the same
    operations.
    A similar argument applies when $π_{i} = r^k$ and $π_{i'} = r^{k'}$.
\end{proof}

Let us define team $A$ to consist of all processes that do
left shifts, and $B$ to consist of all processes that do right shifts.
So for each process $p_i ∈ A$, $π_i = \ell^{k_i}$,
and similarly for each process $p_i ∈ B$, $π_i = r^{k_i}$.

Requiring $n$-discernibility puts some size constraints on the $k_i$:
\begin{lemma}
    \label{lemma-shift-sizes}
    If $R_{A,j} ∩ B_{A,j} = ∅$ for all $j$, then each $k_i ≥ 1$, 
    $∑_{p_i ∈ A} k_i ≤ w-1$, and $∑_{p_i ∈ B} k_i ≤ w-1$.
\end{lemma}
\begin{proof}
    If $k_i = 0$ for some $p_i ∈ A$, then for any $p_{i'} ∈ B$, $q^0
    π_i π_{i'} = q^0 π_{i'}$; the case where $p_i ∈ B$ is symmetric.

    Suppose $k_A = ∑_{p_i ∈ A} k_i ≥ w$. Let $Π_A = π_{i_1} π_{i_2} \dots
    π_{i_α} = \ell^k$ consist of all updates by processes in $A$. Let
    $π_{i'}$ be the update of some $p_{i'} ∈ B$. Then $q^0 Π_A = q^0
    π_{i'} Π_A = 0^w$, and these executions overlap in every operation
    in $Π_A$. The case where $k_B = ∑_{p_i ∈ B} k_i ≥ w$ is
    symmetric.
\end{proof}

We can also exclude certain combinations of left shifts and right
shifts that commute:
\begin{lemma}
    \label{lemma-rllr}
    If $R_{A,j} ∩ B_{A,j} = ∅$ for all $j$, then there do not exist
    disjoint nonempty subsets $A_1, A_2$ of $A$ and $B_1, B_2$ of $B$
    such that
    \begin{align*}
        ∑_{p_i∈A_1} k_i &= ∑_{p_i∈B_1} k_i
        \intertext{and}
        ∑_{p_i∈A_2} k_i &= ∑_{p_i∈B_2} k_i.
    \end{align*}
\end{lemma}
\begin{proof}
    Suppose otherwise. For any set of processes $C$, let $Π_C$ be
    a sequence of updates $π_{i_1} \dots π_{i_α}$ where $i_j$
    enumerates the elements of $C$, and let $k_C = ∑_{p_i∈C} k_i$.
    Write $k$ for $k_{A_1} = k_{B_1}$ and $k'$ for $k_{A_2} =
    k_{B_2}$.

    We can replace the $k$ leftmost and $k'$ rightmost symbols in $q^0$
    with zeros
    by applying $Π_{A_1} Π_{B_1}$ and $Π_{B_2} Π_{A_2}$ in either
    order:
    \begin{align*}
        q^0 \parens*{Π_{A_1} Π_{B_1}} \parens*{Π_{B_2} Π_{A_2}}
        &= \parens*{q^0_{w-1} \dots q^0_0} \parens*{\ell^{k} r^{k}} \parens*{r^{k'} \ell^{k'}}
        \\&= \parens*{q^0_{w-k-1} \dots q^0_0\, 0^k} r^{k} \parens*{r^{k'} \ell^{k'}}
        \\&= \parens*{0^k q^0_{w-k-1} \dots q^0_0} \parens*{r^{k'} \ell^{k'}}
        \\&= \parens*{0^{k+k'} q^0_{w-k-1} \dots q^0_k} \ell^{k'}
        \\&= \parens*{0^k q^0_{w-k-1} \dots q^0_k\, 0^{k'}},
        \intertext{and similarly}
        q^0  \parens*{Π_{B_2} Π_{A_2}} \parens*{Π_{A_1} Π_{B_1}}
        &= \parens*{q^0_{w-1} \dots q^0_0} \parens*{r^{k'} \ell^{k'}} \parens*{\ell^{k} r^{k}}
        \\&= \parens*{q^0_{w-1} \dots q^0_k\, 0^{k'}} \parens*{\ell^{k} r^{k}}
        \\&= \parens*{0^k q^0_{w-k-1} \dots q^0_k\, 0^{k'}}
        \\&= q^0 \parens*{Π_{A_1} Π_{B_1}} \parens*{Π_{B_2} Π_{A_2}}.
    \end{align*}
    Since these two executions start with operations of teams $A$ and
    $B$ respectively, they demonstrate $R_{A,j} ∩ R_{B,j} ≠ ∅$ for any
    $j$ in $A_1∪A_2∪B_1∪B_2$.
\end{proof}

The final step is a counting argument showing that any
collection of $n=w+1$ values $k_i$ that satisfy the constraints of
Lemma~\ref{lemma-shift-sizes} violate the constraints of
Lemma~\ref{lemma-rllr}:
\begin{lemma}
    \label{lemma-logical-shift-register-impossibility}
    A logical shift register with width $w$ is not $n$-discerning for any $n≥w+1$.
\end{lemma}
\begin{proof}
    Suppose that it is $n$-discerning for some $n≥w+1$.
    Then we have a choice of $A$, $B$, and $π_i$ such that
    $R_{A,j} ∩ R_{B,j} = ∅$ for all $j$. From
    Lemmas~\ref{lemma-no-writes} and~\ref{lemma-team-direction}
    we can assume that each $π_i$ for $p_i∈A$ is of the form
    $\ell^{k_i}$ and similarly each $π_i$ for $p_i ∈ B$ is of the form
    $r^{k_i}$. We also have from Lemma~\ref{lemma-shift-sizes} that
    each $k_i$ is at least $1$ and that the sum of the $k_i$ values
    for each team is at most $w-1$. We will show that any choice of
    $k_i$ satisfying these constraints contradicts
    Lemma~\ref{lemma-rllr}.

    Let 
    $A = \Set{p_{i_1}, p_{i_2}, \dots, p_{i_α}}$ 
    and
    $B = \Set{p_{i'_1}, p_{i'_2}, \dots, p_{i'_β}}$.
    Define $L_i = ∑_{j=1}^i k_{i_j}$ and $R_i = ∑_{j=1}^i k_{i'_j}$.
    Let $L = \Set{L_i}$ and $R = \Set{R_i}$.
    Because each $k_i ≥ 1$, all values $L_i$ are distinct, so
    $\card*{L} = \card*{A}$, and similarly $\card*{R} =
    \card*{B}$. Because the sum of the $k_i$ for each set is bounded
    by $w-1$, we also have that $L ∪ R ⊆ \Set{1,\dots,w-1}$, which
    implies $\card*{L∪R} ≤ w-1$.

    Now apply inclusion-exclusion to get $\card*{L∩R} = \card*{L} +
    \card*{R} - \card*{L∪R} ≥ n - (w-1) ≥ 2$ when $n ≥ w+1$.
    Pick two distinct values $L∩R$. These are given by indices $a<a'$ and $b<b'$
    that generate partial sums $L_a = R_b$ and $L_{a'} = R_{b'}$. Let
    \begin{align*}
        A_1 &= \Set{i_1,\dots,i_a} \\
        A_2 &= \Set{i_{a+1},\dots,i_{a'}} \\
        B_1 &= \Set{i_1,\dots,i_b} \\
        B_2 &= \Set{i_{b+1},\dots,i_{b'}}
        \intertext{then}
        ∑_{p_i∈A_1} k_i &= L_a = R_b = ∑_{p_i ∈ B_1} k_i
        \intertext{and}
        ∑_{p_i∈A_2} k_i &= L_{a'} - L_a = R_{b'} - R_b = ∑_{p_i ∈ B_2}
        k_i.
    \end{align*}
    But this is forbidden by Lemma~\ref{lemma-rllr}.
\end{proof}

Next, we argue that a $w$-bit logical shift register is
$w$-discerning. Given the constraints imposed by
Lemmas~\ref{lemma-no-writes} through~\ref{lemma-rllr}, there are not
too many choices to be made here.
\begin{lemma}
    \label{lemma-logical-shift-register-possibility}
    A logical shift register with width $w$ and an alphabet of size
    $2$ or greater is $n$-discerning for any
    $n$ such that $2≤n≤w$.
\end{lemma}
\begin{proof}
    Fix $n$ such that $2≤n≤w$. Assign processes $p_1,\dots,p_{n-1}$ to
    team $A$ and process $p_n$ to team $B$. Let $π_i = r$ for each
    $p_i ∈ A$, and let $π_n = \ell$. Let $q^0 = 1 0^{w-1}$.

    We have two classes of executions depending on which team goes
    first:
    \begin{enumerate}
        \item An execution starting with an operation by team $A$ will
            be of the form $q^0 r^k$ or $q^0 r^k \ell r^{k'}$
            where in either case $k≥1$ and $k+k' ≤ \card{A} = n-1 ≤
            w-1$.
            The fact that $k+k' ≤ w-1$ means that the
            initial $1$ in $q^0_0$ is not shifted off the right end of
            the register; in the second case, the fact that at least
            one $r$ precedes $\ell$ means that it is not shifted off
            the left end either. So an execution with an initial operation by team $A$
            always leaves a nonzero state in the register.
        \item An execution starting with an operation by team $B$ will
            be of the form $q^0 \ell r^k$. Now $q^0 \ell$ immediately
            removes the $1$, leaving $0^w$, and subsequent right shifts
            have no effect. So an execution with an initial operation
            by team $B$ always leaves a zero state in the register.
    \end{enumerate}
    Since the first case only gives nonzero states and the second
    only gives zero states, we have
    $R_{A,j} ∩ R_{B,j} = ∅$ for any $j$, giving $n$-discernability.
\end{proof}

Combining Lemmas~\ref{lemma-logical-shift-register-impossibility}
and~\ref{lemma-logical-shift-register-possibility} with the previous
results of~\cite{Ruppert2000} gives:
\begin{theorem}
    \label{theorem-logical-shift-register}
    The consensus number of a width-$w$ logical shift register is $w$.
\end{theorem}

There is a bit of asymmetry in
Lemmas~\ref{lemma-logical-shift-register-impossibility}
and~\ref{lemma-logical-shift-register-possibility} that is not fully
exepressed in Theorem~\ref{theorem-logical-shift-register}: 
the impossibility result applies to shift registers 
allowing arbitrary shift operations on digits of any
size, while the possibility result uses only single-position
shifts on bits.

\section{Arithmetic shift registers}
\label{section-arithmetic-shift-registers}

An arithmetic shift register supports operations
$\ell^k$ and $s^k$ for any $k≥1$. For this class of registers, the consensus number is
infinite for any $w≥2$:
\begin{theorem}
    \label{theorem-arithmetic-shift-register}
    Let $w≥2$. Then a $w$-bit arithmetic shift register solves
    consensus for any fixed number of processes $n$.
\end{theorem}
\begin{proof}
    We will prove this by showing that such an object is
    $n$-discerning for any $n≥2$. The construction is similar to the
    proof of Lemma~\ref{lemma-logical-shift-register-possibility}, but
    in this case the fact that $s$ preserves the leftmost bit means
    that we don't need to worry about doing too many right shifts.

    Fix $n$ such that $2≤n≤w$. Assign processes $p_1,\dots,p_{n-1}$ to
    team $A$ and process $p_n$ to team $B$. Let $π_i = s$ for each
    $p_i ∈ A$, and let $π_n = \ell$. Let $q^0 = 1 0^{w-1}$.
    Fix some $j$, and consider the set of executions of the form $q^0
    π_{i_1}\dots π_{i_α}$ that include $π_j$.

    Any execution that starts with $q^0 \ell$ immediately reaches a zero
    state that is preserved by any subsequent $s$ operations. So
    $R_{B,j}$ contains only zero states.

    An execution that starts with $q^0 s$ is either of the form $q^0
    s^k$ or $q^0 s^k \ell s^{k'}$ where $k≥1$ in both cases and $k'≥0$
    in the second case. The state following $q^0 s^k$ is $1^{k+1}
    0^{w-k-1}$ when $k < w$ and $1^k$ otherwise. If there is no $\ell$
    operation, this gives a nonzero state. If there is an $\ell$
    operation, it removes at most one of the $\min(k+1,w)≥2$ one bits
    from $q^0 s^k$, leaving $(q^0 s^k \ell)_{w-1} = 1$. The $k'$
    subsequent $s$ operations preserve this one, again giving a
    nonzero state.  Together these two cases show that
    $R_{A,j}$ contains only nonzero states.

    Since $R_{A,j}$ contains only nonzero states and $R_{B,j}$ contains
    only zero states, $R_{A,j} ∩ R_{B,j} = ∅$ as required.
\end{proof}

It is worth noting that while the definition of an arithmetic shift
register allows shifts $\ell^k$ and $s^k$ of arbitrary size, the proof of
Theorem~\ref{theorem-arithmetic-shift-register} only uses the
one-position shift operations $\ell$ and $s$.

The assumption that $w≥2$ is necessary: for a one-bit shift
register, $s$ has no effect, leaving us with only $\Write$ and $\ell$
operations as non-trivial updates. We know from
Lemmas~\ref{lemma-no-writes} and~\ref{lemma-team-direction} that these
are not enough.

\section{Conclusions}
\label{section-conclusions}

We have shown that the consensus number of a register supporting
logical shift operations is equal to its width, giving a
class of familiar shared-memory objects that spans the entire
consensus hierarchy.

This result is similar to the previous results of
Mostefaoui~\etal~\cite{MostefaouiPR2018} and
Ellen~\etal~\cite{EllenGSZ2020} for sliding window
registers/$\ell$-buffers.  The main difference is that these objects
only have one-directional shifts, and rely on processes being able to
choose the new value to be shifted in. Logical shift registers allow
shifting in two directions, which allows them to solve consensus even
though the newly inserted value is fixed. For both classes of objects,
what limits the consensus number to an object's width $w$ is the fact
that doing $w$ shifts erases any previous information.

For arithmetic shift registers, this is no longer the case: sign
extension means that the most significant bit is preserved by
arithmetic right shifts. So, as long as arithmetic shift registers
are wide enough that right shifts are a non-trivial operation, they
can solve consensus for any fixed $n$.

More broadly, the present work complements the program of
Ellen~\etal~\cite{EllenGSZ2020} to compute the consensus power of
registers supporting common machine code instructions subject to
constraints on the number of available objects. Ellen~\etal~consider
registers that support various operations on arbitrary natural
numbers, but their arguments show that many common operations on bounded-size
registers are sufficient to solve consensus for $n$ processes, as long
as the bound is sufficiently large as a function of $n$. A natural
question is which other instructions might give consensus numbers that
scale with the word size.

\bibliographystyle{alpha}
\bibliography{paper.bib}

\begin{thebibliography}{EGSZ20}

\bibitem[EGSZ20]{EllenGSZ2020}
Faith Ellen, Rati Gelashvili, Nir Shavit, and Leqi Zhu.
\newblock A complexity-based classification for multiprocessor synchronization.
\newblock {\em Distributed Computing}, 33(2):125--144, Apr 2020.

\bibitem[FLP85]{FischerLP1985}
Michael~J. Fischer, Nancy~A. Lynch, and Michael~S. Paterson.
\newblock Impossibility of distributed consensus with one faulty process.
\newblock {\em Journal of the ACM}, 32(2):374--382, April 1985.

\bibitem[Her91]{Herlihy1991waitfree}
Maurice Herlihy.
\newblock Wait-free synchronization.
\newblock {\em ACM Trans. Program. Lang. Syst.}, 13(1):124--149, January 1991.

\bibitem[{Int}97]{Intel1997}
{Intel Corporation}.
\newblock {\em Intel Architecture Software Developer's Manual: Volume 2:
  Instruction Set Reference}.
\newblock Intel, 1997.

\bibitem[Jay97]{Jayanti1997}
Prasad Jayanti.
\newblock Robust wait-free hierarchies.
\newblock {\em J. ACM}, 44(4):592--614, 1997.

\bibitem[KR88]{KernighanR1988}
Brian~W Kernighan and Dennis~M Ritchie.
\newblock {\em The {C} programming language}.
\newblock Prentice-Hall, 1988.

\bibitem[LAA87]{LouiA1987}
Michael~C. Loui and Hosame~H. Abu-Amara.
\newblock Memory requirements for agreement among unreliable asynchronous
  processes.
\newblock In Franco~P. Preparata, editor, {\em Parallel and Distributed
  Computing}, volume~4 of {\em Advances in Computing Research}, pages 163--183.
  JAI Press, 1987.

\bibitem[MPR18]{MostefaouiPR2018}
Achour Mostefaoui, Matthieu Perrin, and Michel Raynal.
\newblock A simple object that spans the whole consensus hierarchy.
\newblock {\em Parallel Processing Letters}, 28(02):1850006, 2018.

\bibitem[Rup00]{Ruppert2000}
Eric Ruppert.
\newblock Determining consensus numbers.
\newblock {\em SIAM J. Comput.}, 30(4):1156--1168, 2000.

\end{thebibliography}

\end{document}